\documentclass[12pt]{article}
\usepackage{graphicx}
\usepackage{mdframed}	
\usepackage{amsmath}	
\usepackage{scrextend}
\usepackage{amsthm}
\usepackage{amssymb}
\usepackage{pgfgantt}
\usepackage{algorithm}
\usepackage{algpseudocode}
\usepackage{hyperref}
\usepackage[most]{tcolorbox}
\usepackage{multirow}
\usepackage{multicol}
\usepackage{authblk}
\usepackage[appendix=inline,bibliography=common]{apxproof}

\newtheorem{thm}[equation]{Theorem}
\newtheorem{cor}[equation]{Corollary}
\newtheorem{lem}[equation]{Lemma}

\newtheorem{lp}[equation]{Linear Program}
\newtheorem{defn}[equation]{Definition}

\newtheoremrep{thm}[equation]{Theorem}
\newtheoremrep{lem}[equation]{Lemma}

\theoremstyle{definition}

\theoremstyle{remark}

\begin{document}
\title{The parameterised complexity of generalised temporal domination on temporal graphs with modular structure}

\author[1]{Jessica Enright}
\author[1]{Kitty Meeks}
\author[1]{Elena Moss}
\affil[1]{\small{School of Computing Science, University of Glasgow, Sir Alwyn Williams
Building, Glasgow, G12 8RZ, Scotland, UK}}

\date{}

\maketitle   

\begin{abstract}
Inspired by the static problem $(\alpha,\beta)$-\textsc{Dominating Set}, we propose a general temporal domination problem, called $(\alpha,\beta)$-\textsc{Temporal Dominating Set}. We show that this problem encompasses \textsc{Temporal Dominating Set}, and additionally provides first temporal extensions of problems such as $k$-\textsc{Dominating Set} and $\alpha$-\textsc{Dominating Set}. In this paper, we study the parameterised complexity of $(\alpha,\beta)$-\textsc{Temporal Dominating Set} with respect to \emph{temporal neighbourhood diversity} (TND), \emph{temporal modular-width} (TMW), and \emph{temporal cliquewidth} (TCW). We obtain fixed parameter tractability results for all values of $\alpha$ and $\beta$ with respect to TND; W[1]-hardness with respect to TMW and TCW whenever $\beta$ is in the problem input, or whenever $\alpha\in (0,1)$ and $\beta$ is a fixed constant; and para-NP-hardness with respect to TCW when $\alpha = 0$ and $\beta = 1$, or $\alpha = 1$ and $\beta = 0$. \\

\end{abstract}

\section{Introduction}

Temporal graphs are useful tools to study real-world time-varying networks, such as interpersonal relationships, disease spread, and traffic routing. The volume of research on temporal graph problems has greatly expanded in recent years \cite{casteigtsMeeksMertzios2021,FluschnikMolterNiedermeier2020,holmeSaram2012,michailTemporalGraphs}. Unfortunately, the translation of static graph problems to the temporal world is not straightforward, and many problems which are efficiently solvable on static graphs become intractable on these temporal objects \cite{BumpusMeeks2023,CioniDondiMarinoSchSilva2025,HaagMolterNiedermeierRenken2022}.

Many static graph problems admit multiple temporal versions that capture slightly different properties of interest; here we consider temporal versions of \emph{domination problems}. The classic NP-complete problem \textsc{Dominating Set} (\textsc{ds}) looks for a small set of vertices such that every vertex in a static graph is either in the set, or adjacent to a vertex in the set. There are multiple ways to define this in the temporal setting; see, for example, \cite{casteigtsJourneyDynamicNetworks,casteigtsFlocchini2013,kutnerLJ2024}. We use the definition given by Herrmann \emph{et al.} in \cite{HerrmannKomMora2025}. For convenience, say that a \emph{snapshot} is the static graph induced by all edges which are present at a specific time, and that the \emph{underlying graph} is the static graph induced by edges which are present for at least one snapshot. A \textit{temporal dominating set} is a subset $D$ of vertex-time pairs (called \emph{temporal vertices}) where, for every vertex $u$ in the underlying graph, there exists $(v,t)\in D$ such that $u$ is dominated by $v$ in the snapshot at time $t$. The problem \textsc{Temporal Dominating Set} (\textsc{tds}) is the corresponding decision problem of whether or not there exists such a set of some specified cardinality. 

Lafond and Luo \cite{LafondLuo2023} have recently introduced the problem $(\alpha,\beta)$-\textsc{Dominating Set}, which acts as a natural generalisation of other well-studied domination problems on static graphs. Formally, consider a rational\footnote{The problem is defined for \emph{real numbers} $\alpha\in [0,1]$ in \cite{LafondLuo2023}, however we are primarily interested in the complexity of computation and thus restrict our attention to rational values.} number $\alpha\in [0,1]$, and let $\beta$ be an integer (which can be positive or negative). The problem is defined in \cite{LafondLuo2023} as follows: 

\begin{center}
\begin{tcolorbox}[colback=gray!5!white,colframe=gray!75!black,width=0.975\linewidth]
\noindent\textsc{$(\alpha,\beta)$-Dominating Set} (\textsc{$(\alpha,\beta)$-ds})\\
\textbf{Input:} A graph $G = (V,E)$ and non-negative integer $k$\\
\textbf{Question:} Is there a subset $D\subseteq V$ of size at most $k$ such that, for every $v\in V\setminus D$, $|N(v)\cap D| \geq \alpha |N(v)| + \beta$?
\end{tcolorbox}
\end{center}

For example, when $\alpha = 0$ and $\beta \geq 1$, this is the $k$-\textsc{Dominating Set} problem \cite{chellali2012}; when $\beta=1$, this is the classical version of \textsc{Dominating Set}. When $\beta = 0$ and $\alpha \in (0,1]$, this is the $\alpha$-\textsc{Dominating Set} problem \cite{DunbarHoffman2000}; when $\alpha = 1$ and $\beta \leq 0$, this is the \textsc{Bounded Degree Deletion} problem \cite{ganianKluteOrd2021}, which includes \textsc{Vertex Cover} when $\beta = 0$. 

Because the definition of $(\alpha,\beta)$-\textsc{ds} includes \textsc{ds}, the NP-hardness of the problem is directly inherited. This motivates their appeal in \cite{LafondLuo2023} to \emph{parameterised complexity} \cite{cygan2015,downeyFundamentalsParameterizedComplexity2013}, in which the aim is to find algorithms where the non-polynomial portion of the runtime is restricted to be only in terms of a specific aspect, or \emph{parameter}, of the problem instance. Lafond and Luo study the parameterised complexity of $(\alpha,\beta)$-\textsc{ds} with respect to a range of parameters, specifically those which capture structural aspects of the problem. Among others, this includes \emph{neighbourhood diversity}, \emph{modular-width} and \emph{cliquewidth}. Moreover, they further delineate the problem based on whether or not $\beta$ is given as a part of the input, or a fixed constant. In the former, they show $(\alpha,\beta)$-\textsc{ds} is $W[1]$-hard with respect to modular-width for all $\alpha\in[0,1]$; in the latter, the same complexity follows when $\alpha\in (0,1)$ (see \cite[Corollary 2]{LafondLuo2023}).

To our knowledge, no \emph{temporal} version of this problem has been introduced. We describe a temporal variant of $(\alpha,\beta)$-\textsc{ds} which aims to preserve the spirit and generality of the static version; in particular, we ensure that it properly generalises \textsc{Temporal Dominating Set}. In $(\alpha,\beta)$-\textsc{Temporal Dominating Set} ($(\alpha,\beta)$-\textsc{tds}), we search for a set $D$ of vertex-time pairs such that all vertices which \emph{never} appear in $D$ are adjacent to at least $\alpha|N_{\downarrow}(v)| + \beta$ vertices which appear in $D$ for at least one time; here, $|N_{\downarrow}(v)|$ denotes the number of neighbours of $v$ in the underlying graph. For a formal problem definition, see Section \ref{subsec:formal-prob-defn}.

We provide first parameterised complexity results for $(\alpha,\beta)$-\textsc{tds}, specifically with respect to the temporal extensions of \emph{neighbourhood diversity} ($nd$), \emph{modular-width} ($mw$), and \emph{cliquewidth} ($cw$) introduced by Enright \emph{et al.} in \cite{EnrightHandLJMeeks2024}. The static counterparts were first defined by Lampis \cite{lampis2012}, Gajarsky \textit{et al.} \cite{gajarsky2013}, and Courcelle \textit{et al.} \cite{courcelle1993}, respectively, and have become key structural graph parameters in complexity theory \cite{Ganian2012,hegerfeldKratsch2023,LafondLuo2023}. Intuitively, the link between these parameters is that they each describe a way to group together vertices depending on the similarity of their neighbourhoods. Cliquewdith generalises modular-width, which further generalises neighbourhood diversity.

The temporal versions of these parameters admit an analogous hierarchy. Beginning with the most restrictive, \emph{temporal neighbourhood diversity} (TND) minimises the number of classes needed in a partition of the vertices, such that all vertices in a \emph{block} of the partition have the same neighbourhood at every time. Both \emph{temporal modular-width} (TMW) and \emph{temporal clique-width} (TCW) provide a more general interpretation of this ``block'' structure. TMW relaxes the requirement that edges within a block (in this context known as \emph{modules}) must be \emph{all} present or \emph{all} absent in a given snapshot; TCW gives a further abstraction and is defined as the minimum number of \emph{labels} (corresponding to specific operations) needed to construct a temporal graph. For any temporal graph, we have 
\begin{equation}
\label{ineq:param-hierarchy}
    \text{TND} \geq \text{TMW}\geq \text{TCW}.
\end{equation}An important aspect of this hierarchy is that fpt-results for more general parameters imply such results ``up the chain'': for example, an fpt-algorithm with respect to either TMW or TCW for a problem would imply such a result with respect to TND. For formal definitions of these parameters, see Section \ref{subsec:param-defn}.

As was done for $(\alpha,\beta)$-\textsc{ds} in \cite{LafondLuo2023}, we separately analyse the cases of $(\alpha,\beta)$-\textsc{tds} based on whether $\beta$ is part of the input, or a fixed constant. In Section \ref{sec:fpt-a-b}, we use a novel linear program to show that in either case, and for any choice of $\alpha$, $(\alpha,\beta)$-\textsc{tds} is fixed parameter tractable with respect to TND (see Theorem \ref{thm:tds-fpt-tnd}). In Section \ref{sec:parm-hardness-a-b}, we show the inheritance of certain W[1]-hardness results for $(\alpha,\beta)$-\textsc{tds} from the static problem with respect to TCW and TMW, both whenever $\beta$ is in the input, and when $\beta$ is any constant but $\alpha\in (0,1)$ (see Theorem \ref{thm:tmw-w1-hard}). Finally, we show para-NP-hardness with respect to TCW when $\alpha = 0$ and $\beta=1$, as well as when $\alpha = 1$ and $\beta = 0$ (see Theorem \ref{thm:a-b-p-np-hard}). A summary of these results is presented in Table \ref{tab:resultsTable}.

\begin{figure}[h]
\begin{center}
\begin{tabular}{ |p{2.5cm}||p{1.75cm}|p{1.75cm}|p{1.75cm}||p{1.85cm}|  }
 \hline
 \multicolumn{5}{|c|}{\textbf{Parameterised Complexity Results for $(\alpha,\beta)$-\textsc{tds}}} \\
 \hline
 & $\alpha = 0$ & $\alpha\in (0,1)$ & $\alpha = 1$ & Parameter\\
 \hline
 \multirow{3}{9em}{$\beta$ in the input} &$W[1]$-h &$W[1]$-h &$W[1]$-h & TCW\\
 &$W[1]$-h &$W[1]$-h &$W[1]$-h & TMW\\
 &FPT &FPT &FPT & TND\\
 \hline
  \multirow{3}{9em}{$\beta$ any positive constant} &\textbf{open}* &$W[1]$-h &poly. & TCW\\
 &\textbf{open} &$W[1]$-h &poly. & TMW\\
 &FPT &FPT &poly. & TND\\
 \hline
  \multirow{3}{9em}{$\beta=0$} &poly. &$W[1]$-h &p-NP-h & TCW\\
 &poly. &$W[1]$-h &\textbf{open} & TMW\\
 &poly. &FPT &FPT & TND\\
 \hline
 \multirow{3}{9em}{$\beta$ any negative constant} &poly. &$W[1]$-h &\textbf{open} & TCW\\
 &poly. &$W[1]$-h &\textbf{open} & TMW\\
 &poly. &FPT &FPT & TND\\
 \hline
\end{tabular}
\end{center}
\caption{A summary table of parameterised complexity results obtained for $(\alpha,\beta)$-\textsc{tds}. If $\alpha =1$ and $\beta \geq 1$, then all instances are \emph{no}-instances (unless we are in the degenerate case where $k$ is at least the number of vertices); if $\alpha = 0$ and $\beta \leq 0$, then all instances are \emph{yes}-instances. These cases are hence polynomial-time solvable. *See Theorem \ref{thm:a-b-p-np-hard} for para-NP-hardness in the case when $\beta = 1$.}
\label{tab:resultsTable}
\end{figure}

The key FPT results of this paper answer an open question posed by Herrmann \emph{et al.} in \cite{HerrmannKomMora2025} regarding the fixed parameter tractability of \textsc{tds} with respect to TND. Moreover, the definition of $(\alpha,\beta)$-\textsc{tds} also provides the first temporal translations of problems such as $k$-\textsc{Dominating Set} and \textsc{Bounded Degree Deletion}. As a tool, in Section \ref{sec:ss-probs}, we define a class of problems for which the order and multiplicity of snapshots do not matter; calling them \emph{snapshot-set problems}. We show that $(\alpha,\beta)$-\textsc{tds} is such a problem, and believe that the techniques we use here for snapshot-set problems will be useful for other problems with this property. 

The remainder of this paper is organised as follows. In Section \ref{sec:notation}, we first highlight our notation, before giving formal definitions of the $(\alpha,\beta)$-\textsc{tds} problem in Section \ref{subsec:formal-prob-defn}, and of our considered parameters in Section \ref{subsec:formal-prob-defn}. We give an overview of snapshot-set problems in Section \ref{sec:ss-probs}, before finally, in Sections \ref{sec:fpt-a-b} and \ref{sec:parm-hardness-a-b}, we show parameterised tractability and intractability results for $(\alpha,\beta)$-\textsc{tds}, respectively. We conclude with a summary and open problems in Section \ref{sec:conc}. 

\section{Notation and preliminaries}\label{sec:notation}

Throughout this paper, we use $[N]$ as shorthand for the set $\{1,\dots, N\}$.

A \emph{temporal graph} $\mathcal{G} = (G,\lambda)$ is a static graph $G = (V,E)$ together with a mapping $\lambda: E\to 2^{\mathbb{N}}$ which assigns a set of \emph{active times} to each edge of $G$. The graph $G = (V,E)$ is called the \emph{underlying graph} of $\mathcal{G}$, which we denote by $G_{\downarrow}$. The maximum time for which an edge of $\mathcal{G}$ is active is called the \emph{lifetime}, which we denote by $\Lambda$. The temporal graph $\mathcal{G}$ can also be expressed by a sequence of \emph{snapshots}, each of which is a static graph $G_t$, $t\in [\Lambda]$, such that $G_t = (V, E_t)$, where $E_t$ are the edges active at time $t$; we refer to the pair $(e,t)\in E\times \lambda(e)$ as a \emph{time-edge}. The set of all time-edges in the temporal graph is denoted $\mathcal{E}(\mathcal{G})$. The \emph{temporal vertex set} of $\mathcal{G}$ is the set of vertex-time pairs $V\times [\Lambda]$, where $(v,t)\in V\times [\Lambda]$ refers to the instance of vertex $v$ in snapshot $G_t$. The closed (open) neighbourhood of temporal vertex $(v,t)$ is defined by the closed (open) neighbourhood of $v\in V$ in snapshot $G_t$; we write this $N_{G_t}[v]$ ($N_{G_t}(v)$). Similarly, $N_{\downarrow}[v]$ ($N_{\downarrow}(v)$) is the closed (open) neighbourhood in $G_{\downarrow}$. For the remainder of this paper, we assume that the underlying graph is always connected. In the context of considered problems, this is without loss of generality, since our problems can always be solved on each of the components.

Finally, we define an instance $x$ of a temporal graph problem as a pair $(\mathcal{G},\Gamma)$, where $\mathcal{G}$ is a temporal graph, and $\Gamma$ is a string encoding the rest of the problem instance. Often, $\Gamma$ will often simply be a string of numbers.

\subsection{$(\alpha,\beta)$-\textsc{Temporal Dominating Set}}\label{subsec:formal-prob-defn}

Consider a temporal graph $\mathcal{G}$ with lifetime $\Lambda$ and vertex set $V$. Let $D\subseteq V\times [\Lambda]$ be a subset of temporal vertices. Then, for each vertex $v\in V$ such that for all $t\in [\Lambda]$, $(v,t)\not\in D$, define the \emph{restriction of $D$ at $v$} by $$D(v) := \{u\in V: \exists t\in \{1,\dots, \Lambda\}: (u,t)\in D \text{ and } t\in \lambda(uv)\}.$$Intuitively, $D(v)$ is a subset of $N_{\downarrow}(v)$; however containment can be strict. $D(v)$ aims to capture only vertices $u\in V$ which are both in the underlying neighbourhood of $v$, \emph{and} in $D$ at some time $t$ such that $v\in N_{t}(u)$. We use the restriction to help us define the following general temporal dominating set condition. 

\begin{defn}
    Let $\alpha\in [0,1]\cap \mathbb{Q}$, and $\beta\in\mathbb{Z}$. Say that a set $D\subseteq V\times [\Lambda]$ is an $(\alpha,\beta)$-\emph{temporal dominating set} for $\mathcal{G}$ if, for all $v\in V$ such that $(v,t)\not\in D$, for all $t\in[\Lambda]$, the following inequality holds:$$\left|D(v)\right| \geq \alpha |N_{\downarrow}(v)| + \beta.$$We call this the \emph{$(\alpha,\beta)$-temporal dominating condition}. 
\end{defn}

The corresponding decision problem can be stated as follows: 

\begin{center}
\begin{tcolorbox}[colback=gray!5!white,colframe=gray!75!black,width=0.975\linewidth]
\noindent\textsc{$(\alpha,\beta)$-Temporal Dominating Set} (\textsc{$(\alpha,\beta)$-tds})\\
\textbf{Input:} A temporal graph $\mathcal{G}$ with vertex set $V$ and $k\in\mathbb{Z}_{>0}$\\
\textbf{Question:} Is there a subset $D\subseteq V\times [\Lambda]$ of size at most $k$ which is an $(\alpha,\beta)$-temporal dominating set for $\mathcal{G}$?
\end{tcolorbox}
\end{center}

Notably, this definition generalises \textsc{tds}.

\begin{lemrep}
\label{obs:0-1-is-tds}$(0,1)$-\textsc{tds} is \textsc{Temporal Dominating Set}. 
\end{lemrep}

\begin{proof}
In this case, we require that for all $v\in V$ such that $(v,t)\not\in D$ for all $t\in[\Lambda]$, $|D(v)| \geq 1$. By definition of $D(v)$, for arbitrary vertex $u\in D(v)$, there is some $(u,t)\in D$ which temporally dominates $v$. Moreover, for any vertex $v'$ for which $D(v')$ is not defined, we know there is some time $t'\in[\Lambda]$ such that $(v',t')\in D$. Thus, $v\in V(G_{\downarrow})$ is temporally dominated by itself.
  \end{proof}

The NP-hardness of $(\alpha,\beta)$-\textsc{tds} is an immediate inherited property of \textsc{tds}. We next show that we do not obtain a symmetric equivalence between $(1,0)$-\textsc{tds} and \textsc{Temporal Vertex Cover} (\textsc{tvc}), which asks for a subset $C\subseteq V\times [\Lambda]$ such that every edge in the underlying graph is incident to some $(v,t)\in C$ \cite{akridaTemporalVertexCover2020,HerrmannKomMora2025}.

\begin{lemrep}Every \emph{yes}-instance of \textsc{tvc} is a \emph{yes}-instance of $(1,0)$-\textsc{tds} with the same certificate. However, there exist \emph{yes}-instances of $(1,0)$-\textsc{tds} which are \emph{no}-instances of \textsc{tvc}.
\end{lemrep}
\begin{proof}
Let $D$ be a certificate of $x=(\mathcal{G},k)$ as a \emph{yes}-instance of \textsc{tvc}. If $v\in V$ is disjoint from all temporal vertices in $D$, then we know \emph{all} of its neighbours are in $D$ \emph{at some timestep}, thus $|D(v)|\geq |N_{\downarrow}(v)|$. Since these are the only vertices for which we must check the $(\alpha,\beta)$-temporal dominating condition, we find that $D$ is also a certificate for $x$ as a \emph{yes}-instance of $(1,0)$-\textsc{tds}.

To see that the converse statement does not hold, take, for example, the temporal graph $\mathcal{G}$ on vertex set $V = \{a,b,c\}$ with underlying edgeset $E = \{ab,bc,ca\}$, such that $ab$ is active at time $1$, $bc$ at time $2$ and $ca$ at time $3$. The set $D = \{(b,1), (c,3)\}$ is a certificate that $x=(\mathcal{G},2)$ is a \emph{yes}-instance of $(1,0)$-\textsc{tds}, however this is not a temporal vertex cover.
  \end{proof}

Moreover, the following is a nice property of the $(\alpha,\beta)$-\textsc{tds}. This will be applied in the parameterised hardness proofs of Section \ref{sec:parm-hardness-a-b}. 

\begin{lemrep} 
\label{lem:a-b-property-lemma}Let $x = (\mathcal{G},\Gamma)$ be an arbitrary instance of $(\alpha,\beta)$-\textsc{tds}, where $\alpha,\beta$ and $k$ are encoded in $\Gamma$. If $\mathcal{G}$ has lifetime $\Lambda = 1$, then $x$ is a \emph{yes}-instance of $(\alpha,\beta)$-\textsc{tds} if and only if $\tilde{x} = (G_1,\Gamma)$ is a \emph{yes}-instance of $(\alpha,\beta)$-\textsc{ds}.  
\end{lemrep}

\begin{proof}
    In this setting, the single snapshot $G_1$ is equal to the underlying graph $G_{\downarrow}$. So, for all vertices $v\in V$, $N_{{G}_{\downarrow}}(v) = N_{G_1}(v)$. Consider a subset $D\subseteq V\times \{1\}$, and define $\tilde{D} := \{v\in V: (v,1)\in D\}$. Then $|D| = |\tilde{D}|$, and, whenever there is a vertex $v\in V$ such $(v,1)\not\in D$, we also have $v\not\in \tilde{D}$, by construction. We claim that $D$ is a certificate of $x$ as a \emph{yes}-instance of $(\alpha,\beta)$-\textsc{tds} if and only if $\tilde{D}$ is such for $\tilde{x}$ of $(\alpha,\beta)$-\textsc{ds}. Consider any vertex $v\in V$ such that $(v,1)\not\in D$. Then, $D(v)$ is well-defined, and: 
    \begin{align*}
        w\in D(v) &\iff \exists t: (w,t)\in D, t\in \lambda(wv)\\
        &\iff (w,1)\in D, v\in N_{G_1}(v)\\
        &\iff w\in \tilde{D}, v\in N_{G_1}(v)\\
        &\iff w\in N_{G_1}(v)\cap \tilde{D}.
    \end{align*}
    Since $w\in D(v)$ is arbitrary, it follows that $|N_{G_1}(v)\cap \tilde{D}| = |D(v)|$, and hence $|N_{G_1}(v)\cap \tilde{D}|\geq \alpha|N_{G_1}(v)|+\beta$ if and only if $|D(v)| \geq \alpha|N_{\downarrow}(v)| + \beta$. Thus, both $D$ and $\tilde{D}$ certify their respective instances.
  \end{proof}

\subsection{Formal parameter definitions}\label{subsec:param-defn}

We are interested in three temporal parameters: \emph{temporal neighbourhood diversity} (\emph{TND}), \emph{temporal modular-width} (\emph{TMW}), and \emph{temporal cliquewidth} (\emph{TCW}). Here, we provide their formal definitions. We begin with the most restrictive, which is the focus of Section \ref{sec:fpt-a-b}.

\begin{defn}{\cite[Definition 8]{EnrightHandLJMeeks2024}}
    Consider a temporal graph $\mathcal{G} = (G,\lambda)$ with lifetime $\Lambda$. We say that two vertices $v,u \in V$ have the same \emph{temporal type} if, for all times $t\in \{1,\dots,\Lambda\}$, $$N_{G_t}[v]\setminus\{u,v\} = N_{G_t}[u]\setminus\{u,v\}.$$ A \emph{temporal neighbourhood partition} of $\mathcal{G}$ is a partition $\Theta = \{B_1, \dots, B_{\phi}\}\subseteq 2^{V}$ of the vertex set of the underlying graph $G = (V,E)$ such that any two vertices in the same part have the same temporal type. We call the parts $B_1, \dots, B_{\phi}$ \emph{blocks}. The \emph{temporal neighbourhood diversity} (\emph{TND}) of $\mathcal{G}$ is the smallest $\phi\in\mathbb{N}$ for which a temporal neighbourhood partition into $\phi$ blocks exists. 
\end{defn}

As observed by Enright \textit{et al.} in \cite[Lemma 9]{EnrightHandLJMeeks2024}, in each snapshot of $\mathcal{G}$, all blocks of a temporal neighbourhood partition induce subgraphs which are either independent sets or cliques. Between each block in a snapshot, the edges are either \emph{all present} (and hence the blocks are \emph{complete to each other}), or \emph{all absent} (blocks are independent of each other). 

In Section \ref{sec:parm-hardness-a-b}, we look to the parameters temporal modular-width and temporal cliquewidth, as defined by Enright \emph{et al.} in \cite{EnrightHandLJMeeks2024}. Both TMW and TCW provide a more general interpretation of the ``block'' structure imposed by a temporal neighbourhood partition. TMW generalizes TND by relaxing the requirement that each block induces a clique or an independent set, in every snapshot.

\begin{defn}{\cite[Definition 6]{EnrightHandLJMeeks2024}}
\label{defn:TMW}
    Suppose a temporal graph $\mathcal{G} = ((V,E),\lambda)$ of lifetime $\Lambda$ can be constructed by an algebraic expression $A$ which uses the following operations: 
    \begin{enumerate}
        \item[T1] Creating an isolated vertex 
        \item[T2] The substitution of temporal graphs $\mathcal{G}_1, \dots, \mathcal{G}_n$ (\emph{modules}) into a temporal graph $\mathcal{G}'$ with vertices $v_1, \dots, v_n$, giving the graph $\mathcal{G}'(\mathcal{G}_1,\dots, \mathcal{G}_n)$ with vertex set $\bigcup_{i\in [n]} V(\mathcal{G}_i)$ and time-edge set $\bigcup_{i\in [n]}\mathcal{E}(\mathcal{G}_i) \cup \{(vw,t): v\in V(\mathcal{G}_i), w\in V(\mathcal{G}_j), (v_iv_j,t)\in \mathcal{E}(\mathcal{G}')\}.$
    \end{enumerate}
    The width of an expression $A$ is the maximum number of operands in an occurrence of T2 in $A$. The \emph{temporal modular width} (\emph{TMW}) of $\mathcal{G}$ is the minimum width of an expression $A$ which constructs $\mathcal{G}$; denote this by $\mu(\mathcal{G})$, or simply $\mu$, when the context is clear.
\end{defn}

Finally, temporal clique-width gives a further level of abstraction and is defined as the minimum number of \emph{labels} required to construct a temporal graph. Such labels correspond to the operations: create a vertex with a new label, take the disjoint union of two graphs, take the complete join of vertices between specified labels, and relabel all vertices of one label to another. 

\section{Snapshot-set problems}
\label{sec:ss-probs}

Let $\mathcal{G} = (G,\lambda)$ be a temporal graph with lifetime $\Lambda$. Define the \emph{set of unique snapshots} of $\mathcal{G}$ as $$\mathcal{S}(\mathcal{G}) := \{G_1, \dots, G_{\tilde{\Lambda}}\},$$where, whenever $G_i,G_j\in \mathcal{S}(\mathcal{G})$ and $i\neq j$, $G_i\neq G_j$, and, for all snapshots $G_t$, $t\in [\Lambda]$, there exists some $G_l\in\mathcal{S}(\mathcal{G})$ such that $G_t = G_l$. Importantly, this is \emph{not a multiset} and hence we may have $\tilde{\Lambda} < \Lambda$, when snapshots of $\mathcal{G}$ are repeated. Moreover, we will assume that, for each set of snapshots $\{G_{i_1},\dots, G_{i_l}\}$ such that $G_{i_1} = \dots = G_{i_l}$, the \emph{representative snapshot} which we include in $\mathcal{S}(\mathcal{G})$ is $G_{i_1}$: that which minimises the index. Define the mapping $\tau_{\mathcal{G}}$ from $[\Lambda]$ to the set of representative snapshots' times. In particular, $\tau_{\mathcal{G}}^{-1}(t)$ is the set of all times for which the representative snapshot is $G_t$. 

For temporal graphs $\mathcal{G}$ and $\mathcal{G}'$, say that $\mathcal{S}(\mathcal{G}) = \mathcal{S}(\mathcal{G}')$ if $G_{\downarrow} = G'_{\downarrow}$, and there exists a bijection $\sigma$ between the times of representative snapshots in $\mathcal{S}(\mathcal{G})$ and $\mathcal{S}(\mathcal{G}')$, such that $\sigma(t) = t'$ if and only if $G_t = G_{t'}$, for $G_t\in\mathcal{S}(\mathcal{G})$ and $G_{t'}\in\mathcal{S}(\mathcal{G}')$. The following definition describes a family of problems for which it is sufficient to exclusively look at the set of unique snapshots. 

\begin{defn}
\label{def:snapshot-unique}
    Let $P$ be a temporal graph problem, and consider any instances $x = (\mathcal{G},\Gamma)$ and $x' = (\mathcal{G}',\Gamma)$ such that $\mathcal{S}(\mathcal{G}) = \mathcal{S}(\mathcal{G}')$. Say $P$ a \emph{snapshot-set problem} if we have that $x$ is a \emph{yes}-instance of $P$ if and only if $x'$ is a \emph{yes}-instance of $P$.
\end{defn}

Such problems have some nice properties. 

\begin{lemrep} 
\label{lem:lifetime_bound}Let $P$ be a snapshot-set problem, $x = (\mathcal{G},\Gamma)$ be an instance such that $\mathcal{G}$ is a temporal graph with lifetime $\Lambda$ and TND $\phi$. Then $|\mathcal{S}(\mathcal{G})|\leq 2^{\phi^2}.$
\end{lemrep}

\begin{proof}
    Due to the strict temporal edge structure imposed by a minimum size temporal neighbourhood partition $\Theta$, the number of \emph{possible} unique snapshots can be found by considering the number of possible edge connections both \emph{between} and \emph{within} blocks. This then upper-bounds the size of $\mathcal{S}(\mathcal{G})$. Between the $\phi$ blocks, there are ${\phi \choose 2}$ possible sets of edges present, yielding $2^{\phi\choose 2}$ unique snapshots based on these edge sets alone. Moreover, each block can be either a clique or an independent set at each timestep, giving $2^{\phi}$ possible snapshots. Compounding these, we find that the number of unique snapshots is $2^{{\phi \choose 2}+\phi}\leq 2^{\phi^2}$, which is the desired bound. 
  \end{proof}

\begin{lemrep} 
\label{lem:equiv-instance-ss-problem}Let $x = (\mathcal{G},\Gamma)$ be an instance of snapshot-set problem $P$, where $\mathcal{G}$ has lifetime $\Lambda$, TND $\phi$, and $n$ vertices. In time $O\left(2^{\phi^2}\cdot n^2\cdot \Lambda\right)$ we can compute an instance $\tilde{x} = (\tilde{\mathcal{G}},\Gamma)$, called the \emph{reduced instance}, where $\tilde{\mathcal{G}}$ has lifetime $\tilde{\Lambda} := |\mathcal{S}(\mathcal{G})|$, such that $x$ is a \emph{yes}-instance of $P$ if and only if $\tilde{x}$ is a \emph{yes}-instance of $P$.
\end{lemrep}

\begin{proof}
    To compute the instance $\tilde{x}$, we can store the $2^{\phi^2}$ possible snapshots (Lemma \ref{lem:lifetime_bound}) in a list. We then iterate through all $\Lambda$ snapshots in $\mathcal{G}$. To check if a given snapshot is equal to one already in the list takes time at most $n^2$. Putting this together yields a runtime of $$O\left(2^{\phi^2}\cdot n^2\cdot \Lambda\right).$$Observe that $\mathcal{S}(\mathcal{G}) = \mathcal{S}(\tilde{\mathcal{G}})$. Hence, since $P$ is a snapshot-set problem, $x$ is a \emph{yes}-instance of $P$ if and only if $\tilde{x}$ is a \emph{yes}-instance of $P$.
  \end{proof}

We finally show that $(\alpha,\beta)$-\textsc{tds} is a snapshot-set problem. In the next section, we will leverage this fact to apply Lemmas \ref{lem:lifetime_bound} and \ref{lem:equiv-instance-ss-problem}.

\begin{lemrep} 
\label{lem:a-b-tds-snapshot-set}$(\alpha,\beta)$-\textsc{tds} is a snapshot-set problem. 
\end{lemrep}

\begin{proof}
    Consider any two instances $x = (\mathcal{G} = (G,\lambda),\Gamma)$ and $\tilde{x} = (\tilde{\mathcal{G}}=(G,\tilde{\lambda}),\Gamma)$, such that $\mathcal{S}(\mathcal{G}) = \mathcal{S}(\tilde{\mathcal{G}})$ (this implies $V(\mathcal{G}) = V(\tilde{\mathcal{G}})=:V$), $k$ is encoded in $\Gamma$, and suppose without loss of generality that $\tilde{\mathcal{G}}$ has lifetime $\tilde{\Lambda} = |\mathcal{S}(\tilde{\mathcal{G}})|$. Recall that $\sigma$ is the bijection between the times of the representative snapshots in $\mathcal{S}(\mathcal{G})$ and $\mathcal{S}(\tilde{\mathcal{G}})$.
    
    Now, let $D$ be a certificate for $x$ as a \emph{yes}-instance. For each $(u,t)\in D$, recall that $(u,\tau_{\mathcal{G}}(t))$ is the corresponding temporal vertex in the \emph{representative snapshot}. Using this, we construct the set $\tilde{D}$ such that, for all $(u,t)\in D$, we add $(u,\sigma(\tau_{\mathcal{G}}(t)))$ to $\tilde{D}$. We claim that $\tilde{D}$ is a \emph{yes}-certificate for $\tilde{x}$. 

    Consider an arbitrary vertex $v\in V$ for which $(v,t)\not\in D$, for every $t\in [\Lambda]$. Since the set of representative times is a subset of $[\Lambda]$, this in turn means that $(v,\sigma(\tau_{\mathcal{G}}(t)))\not\in \tilde{D}$, by construction, and hence $\tilde{D}(v)$ is well-defined. Observe the following:
    \begin{align*}
        \forall w\in D(v) &\implies \exists \sigma^{-1}(t): (w,\sigma^{-1}(t))\in D: \sigma^{-1}(t)\in \lambda(wv)\\
        &\implies \exists t: (w,t)\in \tilde{D}:t\in\tilde{\lambda}(wv)\\
        &\implies w\in \tilde{D}(v).
    \end{align*}
    and hence $|\tilde{D}(v)|\geq |D(v)|$ for all such vertices. In particular, $$|\tilde{D}(v)|\geq \alpha|N_{G_{\downarrow}}(v)| + \beta = \alpha|N_{\tilde{G}_{\downarrow}}(v)| + \beta.$$Because $|\tilde{D}| \leq |D|\leq k$ by construction, and since $v\in V$ was arbitrary such that $(v,t)\not\in D$ for all $t\in [\Lambda]$, it follows that $\tilde{D}$ is a valid certificate for $\tilde{x}$. 

    On the other hand, suppose $\tilde{D}'$ certifies $\tilde{x}$ as a \emph{yes}-instance. Construct $D'$ by adding $(v,\sigma^{-1}(t))$ whenever $(v,t)\in \tilde{D}'$. For each $v\in V$ such that $(v,t)\not\in \tilde{D}'$ for all $t\in [\Lambda]$, it follows that $(v,\sigma^{-1}(t))\not\in D'$, and hence $D'(v)$ is well-defined. In particular, 
    \begin{align*}
        \forall w\in \tilde{D}'(v) &\implies \exists t': (w,t')\in \tilde{D}', t'\in\tilde{\lambda}(wv)\\
        &\implies \exists t\in\tau_{\mathcal{G}}^{-1}(\sigma^{-1}(t')): (w,t)\in D', t\in \lambda(wv)\\
        &\implies w\in D'(v),
    \end{align*}
    The result follows symmetrically to the previous case.
  \end{proof}

\section{Parameterised tractability of $(\alpha,\beta)$-TDS}
\label{sec:fpt-a-b}

The goal of this section is to prove the following.

\begin{thm}
\label{thm:tds-fpt-tnd}
    For all $\alpha\in[0,1]$, $(\alpha,\beta)$-\textsc{tds} is in FPT with respect to TND, even if $\beta$ is part of the input. 
\end{thm}

To prove Theorem \ref{thm:tds-fpt-tnd}, we first define an integer linear program (ILP) to represent the $(\alpha,\beta)$-\textsc{tds} problem. Consider an instance $(\mathcal{G},\Gamma)$ of $(\alpha,\beta)$-\textsc{tds}, where $\alpha,\beta,k\in \Gamma$, such that $\mathcal{G}$ has TND $\phi$ (with corresponding temporal neighbourhood partition $\Theta = \{B_1, \dots, B_{\phi}\}$), and lifetime $\Lambda$. Let $\mathcal{I}\in 2^{[\phi]}$. Intuitively,  $\mathcal{I}$ is an index set which represents a ``guess'' of blocks in $\Theta$ with \emph{at least one} vertex not included in the temporal dominating set, at any time. We then analyse such blocks against the $(\alpha,\beta)$-temporal dominating condition.

For each $i\in [\phi]$ and each $S\in 2^{[\Lambda]}$, introduce the ILP variable $x_{i,S}\in\mathbb{Z}_{\geq 0}$ to represent the number of vertices in block $B_i$ which are selected \emph{exactly} at times in the set $S$. Let $\gamma_i$ represent the size of the open underlying neighbourhood of an arbitrary vertex in $B_i$. This allows us to encode the expression $\alpha |N_{\downarrow}(v)| + \beta$ in the ILP. Finally, if $\lambda$ is the edge labelling function of $\mathcal{G}$, let the shorthand $\lambda(ij)$ denote the set of times when blocks $B_i$ and $B_j$ are adjacent, since the set of times for which edges between vertices of the two blocks will be identical for all vertices within the blocks. Let $\lambda(ii)$ be the set of times when block $i$ is a clique. Consider the following integer linear program: 

\begin{lp}[ILP$(\mathcal{I},\vec{\gamma}\in\mathbb{Z}_{\geq 0}^{\phi})$]
\label{lp:tnd-tds}
    \begin{align}
        x_{i,\emptyset} > 0&\quad\forall i\in \mathcal{I}\label{lp:2}\\
        x_{i,\emptyset} = 0&\quad\forall i\not\in\mathcal{I}\label{lp:3}\\
        \sum_{S\in 2^{[\Lambda]}}x_{i,S}= |B_i|&\quad\forall i\in [\phi]\label{lp:4}\\
        \sum_{i\in [\phi]}\sum_{S\in 2^{[\Lambda]}} |S|\cdot x_{i,S}\leq k&\label{lp:5}\\
        \sum_{j\in [\phi]}\sum_{\substack{
            S\in 2^{[\Lambda]}:\\
            \lambda(ij)\cap S\neq \emptyset}} x_{j,S}\geq \alpha\gamma_i + \beta&\quad\forall i\in \mathcal{I}\label{lp:6}
    \end{align}
\end{lp}

Algorithm \ref{alg:a-b-tds} provides the general workflow to solve $(\alpha,\beta)$-\textsc{tds}, and executes ILP$(\mathcal{I},\vec{\gamma})$ on a sequence of candidate solutions.

\begin{algorithm}
    \caption{$(\alpha,\beta)$-\textsc{tds}$(x = (\mathcal{G},\Gamma))$}\label{alg:a-b-tds}
    \begin{algorithmic}
    \State{$\Theta = \{B_1,\dots, B_{\phi}\}$ is the minimum temporal neighbourhood partition of $\mathcal{G}$; $\gamma\in\mathbb{Z}^{\phi}_{\geq 0}$ is defined such that $\gamma_i = |N_{\downarrow}(v)|$ for all $v\in B_i$}
        \For{$\mathcal{I}\in 2^{[\phi]}$}
            \If{ILP$(\mathcal{I},\vec{\gamma})$ is \emph{feasible}}
                \State{\Return \textbf{True}}
            \EndIf
        \EndFor
    \State{\Return \textbf{False}}
    \end{algorithmic}
\end{algorithm}

The next lemma shows that Linear Program \ref{lp:tnd-tds} properly encodes the problem. 

\begin{lem}
\label{lem:ilp-correctness}
    Let $(\mathcal{G},\Gamma)$ be an instance of $(\alpha,\beta)$-\textsc{tds}, where $\alpha,\beta,k\in \Gamma$, such that $\mathcal{G}$ has TND $\phi$, with corresponding temporal neighbourhood partition $\Theta = \{B_1, \dots, B_{\phi}\}$, and lifetime $\Lambda$. Then, $(\mathcal{G},\Gamma)$ is a \emph{yes}-instance of $(\alpha,\beta)$-\textsc{tds} if and only if there exists $\mathcal{I}\in 2^{[\phi]}$ such that ILP$(\mathcal{I},\vec{\gamma})$ is feasible.
\end{lem}

\begin{proof}
    Suppose first that ILP$(\mathcal{I},\gamma)$ is feasible with certificate $X$. Let $D\subseteq V\times [\Lambda]$ be constructed such that each block $B_i$ contributes $x_{i,S}$ temporal vertices to $D$ at exactly the times in $S$. We claim that $D$ is a valid certificate for $(\mathcal{G},\Gamma)$ as a \emph{yes}-instance of $(\alpha,\beta)$-\textsc{tds}. 

    By constraint \ref{lp:4} of Linear Program \ref{lp:tnd-tds}, there are always enough remaining vertices to be included in $D$ at exactly the prescribed times. Constraints \ref{lp:2} and \ref{lp:3} ensure the set $\mathcal{I}$ indeed contains exactly the indices of blocks which have a positive number of vertices that are \emph{never} included in $D$, and hence are the vertices that must satisfy the $(\alpha,\beta)$-temporal dominating condition. By constraint \ref{lp:5}, the size of $D$ is at most $k$. 
    
    Consider now an arbitrary vertex $v\in B_i$ such that $i\in \mathcal{I}$, and $v$ is one of the $x_{i,\emptyset}$ vertices \emph{not} included in $D$ at any time. Then, for arbitrary $B_j\in \Theta$, observe:
    \begin{align*}
        |\{u&\in B_j:\exists t, (u,t)\in D, t\in\lambda(ij)\}|\\
        &= \sum_{S\in 2^{[\Lambda]}}|\{v\in B_j: \forall t\in S, (v,t)\in D,\forall t'\not\in S, (v,t')\not\in D, S\cap \lambda(ij)\neq \emptyset\}|\\
        &=\sum_{\substack{
            S\in 2^{[\Lambda]}:\\
            \lambda(ij)\cap S\neq \emptyset}} x_{j,S}
    \end{align*}and so
    \begin{align*}
        |D(v)|&=|\{u\in V:\exists t, \exists j,(u,t)\in D\cap (B_j\times\{t\}), t\in\lambda(ij)\}|\\
        &=\sum_{j\in[\phi]}\sum_{\substack{
            S\in 2^{[\Lambda]}:\\
            \lambda(ij)\cap S\neq \emptyset}} x_{j,S}.
    \end{align*}
    By constraint \ref{lp:6}, we know that this is at least $\alpha\gamma_i + \beta$ for all $i\in\mathcal{I}$, which, as discussed, is exactly the ILP encoding of $\alpha|N_{\downarrow}(v)|+\beta$. Thus, $D$ indeed certifies $(\mathcal{G},\Gamma)$ as a \emph{yes}-instance of $(\alpha,\beta)$-\textsc{tds}. 

    On the other hand, suppose $D'$ is a certificate for $(\mathcal{G},\Gamma)$. Construct the following input to Linear Program \ref{lp:tnd-tds}. Let $\mathcal{I}$ be a subset of $[\phi]$ corresponding to the index set of blocks in $\Theta$ containing a vertex \emph{not in $D'$}, for all times $t\in [\Lambda]$. Subsequently, for all $i\in\mathcal{I}$, define $x_{i,\emptyset}'$ to be the number of such vertices in block $B_i$; for all $i\not\in \mathcal{I}$, set $x_{i,\emptyset}' = 0$. For all other non-empty $S\in 2^{[\Lambda]}$, set $x_{i,S}'$ to be the number of vertices in block $B_i$ included in $D'$ at \emph{exactly} the times in $S$. We claim $X' = [x_{i,S}']_{i\in [\phi], S\in 2^{[\Lambda]}}$ certifies that $ILP(\mathcal{I},\gamma)$ is feasible. 

    Constraints \ref{lp:2} and \ref{lp:3} follow by construction of the set $\mathcal{I}$, and constraint \ref{lp:4} is satisfied since no vertex from a single block can be added for \emph{exactly} two distinct sets of times, $S\neq S'$. Constraint \ref{lp:5} again follows since $|D'|\leq k$ (note that if $S = \emptyset$, $|S| = 0$). Finally, consider arbitrary $i\in\mathcal{I}$, and any vertex $u\in B_i$. Let $w\in B_i$ be any vertex such that $(w,t')\not\in D'$ for all times $t'\in[\Lambda]$. Then, 
    \begin{align*}
        \sum_{j\in [\phi]}\sum_{\substack{
            S\in 2^{[\Lambda]}:\\
            \lambda(ij)\cap S\neq \emptyset}} x_{j,S}' =& \sum_{j\in [\phi]}\sum_{\substack{
            S\in 2^{[\Lambda]}:\\
            \lambda(ij)\cap S\neq \emptyset}} |\{v\in B_j: s(v) = S\}|\\
            =& \sum_{j\in [\phi]}|\{v\in B_j:\exists t\in[\Lambda], (v,t)\in D', t\in\lambda(ij)\}|\\
            \geq&
            |\{v\in B_i:\exists t\in[\Lambda], (v,t)\in D', t\in\lambda(uv)\}| \\&+\sum_{j\in [\phi]\setminus \{i\}}|\{v\in B_j:\exists t\in[\Lambda], (v,t)\in D', t\in\lambda(ij)\}|\\
            =& |\{v\in V:\exists t\in[\Lambda], (v,t)\in D', t\in\lambda(uv)\}|.
    \end{align*}
    where the inequality in line 3 captures that, if there is some $t\in [\Lambda]$ for which $(u,t)\in D'$, the sum in line 2 will count $u$, but the sum in line 3 will not. Since $w,u\in B_i$, both vertices have the same temporal type. Thus, we obtain 
    \begin{align*}
        &|\{v\in V:\exists t\in[\Lambda], (v,t)\in D', t\in\lambda(uv)\}|\\&= |\{v\in V:\exists t\in[\Lambda], (v,t)\in D', t\in\lambda(wv)\}|\\
    &= |D'(w)|.
    \end{align*}

    Since $D'$ certifies $(\mathcal{G},\Gamma)$, for all $i\in \mathcal{I}$ we have $$\sum_{j\in [\phi]}\sum_{\substack{
            S\in 2^{[\Lambda]}:\\
            \lambda(ij)\cap S\neq \emptyset}} x_{j,S}'\geq |D'(w)| \geq \alpha |N_{\downarrow}(w)|+\beta =\alpha\gamma_i + \beta,$$and so condition \ref{lp:6} is satisfied, and the result follows.
  \end{proof}

From this, the correctness of Algorithm \ref{alg:a-b-tds} is an immediate corollary.

\begin{cor}[Correctness of Algorithm \ref{alg:a-b-tds}]
\label{lem:correctness-alg}
    Let $x=(\mathcal{G},\Gamma)$ be an instance of $(\alpha,\beta)$-\textsc{tds}, where $\alpha,\beta,k\in \Gamma$, such that $\mathcal{G}$ has temporal neighbourhood diversity $\phi$, with corresponding temporal neighbourhood partition $\Theta = \{B_1, \dots, B_{\phi}\}$, and lifetime $\Lambda$. Then, Algorithm \ref{alg:a-b-tds} returns \textbf{True} if and only if $x$ is a \emph{yes}-instance of $(\alpha,\beta)$-\textsc{tds}.
\end{cor}

To evaluate the runtime of Algorithm \ref{alg:a-b-tds} in the proof of Theorem \ref{thm:tds-fpt-tnd}, we use the following result \cite{frankTardos1987,lenstra1983}. 

\begin{thm}[\cite{LafondLuo2023}, Theorem 4]
\label{thm:ilp-runtime}
    An integer linear program on $n$ variables can be solved in time $O(n^{2.5n + O(n)}\cdot|x|)$, where $|x|$ is the size of the instance.
\end{thm}

We are finally ready to prove the main result of this section.

\begin{proof}[Theorem \ref{thm:tds-fpt-tnd}]
    The correctness of Algorithm \ref{alg:a-b-tds} is established by Corollary \ref{lem:correctness-alg}. By Lemma \ref{lem:equiv-instance-ss-problem}, recall that from any instance $x = (\mathcal{G},\Gamma)$, where $\mathcal{G}$ has lifetime $\Lambda$, $n$ vertices, and TND $\phi$, we can obtain the \emph{reduced instance} $\tilde{x} = (\tilde{\mathcal{G}},\Gamma)$ where $\tilde{G}$ has lifetime $\tilde{\Lambda} \leq 2^{\phi^2}$, in time at most $O(2^{\phi^2}\cdot n^2 \cdot \Lambda)$. 
    
    On $\tilde{x}$, Algorithm \ref{alg:a-b-tds} executes Linear Program \ref{lp:tnd-tds} at most $2^{\phi}$ times. Moreover, for arbitrary choice of set $\mathcal{I}\in 2^{[\phi]}$, we know that the number of variables in ILP$(\mathcal{I},\vec{\gamma})$ is at most $|\mathcal{I}|\cdot 2^{\tilde{\Lambda}} \leq 2^{\phi}\cdot 2^{\phi^2}$ on instance $\tilde{x}$. Because this is exponential only in the TND of $\tilde{\mathcal{G}}$, Theorem \ref{thm:ilp-runtime} gives an fpt-runtime for solving $(\alpha,\beta)$-\textsc{tds} on instance $\tilde{x}$. Note that including the pre-processing step to obtain $\tilde{x}$ from $x$ does not change the asymptotic bound. Hence, we conclude that Algorithm \ref{alg:a-b-tds} is an fpt-algorithm with respect to TND, and that $(\alpha,\beta)$-\textsc{tds} is in FPT.
  \end{proof}

\section{Parameterised hardness of $(\alpha,\beta)$-TDS}
\label{sec:parm-hardness-a-b}

The first goal of this section is to prove the following theorem.

\begin{thm}
\label{thm:tmw-w1-hard}
    For $\alpha\in [0,1]$ and $\beta$ in the input, $(\alpha,\beta)$-\textsc{tds} is W[1]-hard with respect to TMW. Moreover, if $\alpha\in (0,1)$, $(\alpha,\beta)$-\textsc{tds} is W[1]-hard with respect to TMW \emph{even if} $\beta$ is a constant.
\end{thm}

\begin{proof}
    We proceed via reduction from $(\alpha,\beta)$-\textsc{ds}. For an arbitrary instance $x = (G,\Gamma)$ of $(\alpha,\beta)$-\textsc{ds} (where $\alpha,\beta,k$ are encoded in $\Gamma$), we consider the instance $\tilde{x} = (\tilde{\mathcal{G}},\Gamma)$ of $(\alpha,\beta)$-\textsc{tds} such that $\tilde{\mathcal{G}}$ has lifetime $1$, and $\tilde{G}_{\downarrow} = \tilde{G}_1 = G$. In \cite{EnrightHandLJMeeks2024}, Enright \emph{et al.} note that, when all edges are active at the same times, the TMW of a temporal graph is equal to the (\emph{static}) modular-width of its underlying graph. In the case of $\tilde{x}$, we find $$\mu(\tilde{\mathcal{G}}) = mw(\tilde{G}_1) = mw(G).$$By Lemma \ref{lem:a-b-property-lemma} we know that $\tilde{x}$ is a \emph{yes}-instance of $(\alpha,\beta)$-\textsc{tds} if and only if $x$ is a \emph{yes}-instance of $(\alpha,\beta)$-\textsc{ds}. Thus, the result follows by \cite[Corollary 2]{LafondLuo2023}.
  \end{proof}

By inequality \ref{ineq:param-hierarchy}, we get the following corollary. 

\begin{cor}
    For $\alpha\in [0,1]$ and $\beta$ in the input, $(\alpha,\beta)$-\textsc{tds} is W[1]-hard with respect to TCW. Moreover, if $\alpha\in (0,1)$, $(\alpha,\beta)$-\textsc{tds} is W[1]-hard with respect to TCW\emph{even if} $\beta$ is a constant.
\end{cor}

Our remaining goal is to show the following: 

\begin{thm}
\label{thm:a-b-p-np-hard}
    $(\alpha,\beta)$-\textsc{tds} is para-NP-hard with respect to TCW when: $\alpha = 0$ and $\beta =1$, or $\alpha = 1$ and $\beta = 0$.
\end{thm}

Since $(0,1)$-\textsc{tds} is \textsc{tds}, the problem is NP-hard even when the underlying graph is a star \cite[Theorem 3.1]{HerrmannKomMora2025}. We show the same holds for $(1,0)$-\textsc{tds}.

\begin{lemrep} 
\label{lem:1-0-tds-np-hard}$(1,0)$-\textsc{tds} is NP-hard even when the underlying graph is a star. 
\end{lemrep}

\begin{proof}
    This follows almost identically to the NP-hardness proof of \textsc{tvc} by Akrida \emph{et al.} \cite{akridaTemporalVertexCover2020}. We reduce from \textsc{Set Cover}, defined as follows. 

    \begin{center}
\begin{tcolorbox}[colback=gray!5!white,colframe=gray!75!black,width=0.975\linewidth]
\noindent\textsc{Set Cover} (\textsc{sc})\\
\textbf{Input:} A universe $U = [n]$, a collection $\mathcal{C}=\{C_1, \dots, C_m\}$ of subsets of $U$ such that $\cup_{C_i\in\mathcal{C}} C_i = U$ and an integer $k$.\\
\textbf{Question:} Is there a subset $\mathcal{C}'\subseteq \mathcal{C}$ such that $|\mathcal{C}'|\leq k$ such that $\cup_{C_i\in\mathcal{C}'}C_i = U$?
\end{tcolorbox}
\end{center}
    
    In particular, given an instance $(U,\mathcal{C},k)$ of \textsc{sc}, construct $\mathcal{G} = (G,\lambda)$ such that $\mathcal{G}$ has lifetime $\Lambda = m$, $G$ is a star on $n+1$ vertices: $$V(G) = \{c, v_1, \dots, v_n\},$$and the mapping $\lambda$ is defined such that, for all times $t\in [m]$, $G_t$ has edges between $c$ and $v_j$ such that $j\in C_t$. 

    ($\implies$): Let $D$ be a minimum size certificate for $x=(\mathcal{G},\Gamma)$ as a \emph{yes}-instance of $(1,0)$-\textsc{tds}, where $k$ is in $\Gamma$. Let $D_t = \{(v,t)\in D: v\in V(G)\}$. Because $G_{\downarrow}$ is a star, and $D$ is minimal, we can assume without loss of generality that for every $t\in [m]$, $D_t = \{(c,t)\}$ or $D_t = \emptyset$. We claim the collection $\mathcal{C}' = \{C_t\in\mathcal{C}: D_t\neq \emptyset\}$ is a \emph{yes}-certificate for $(U,\mathcal{C})$. If $D_t\neq \emptyset$, then $(c,t)\in D$ dominates all vertices $v_j$, $j\in C_t$. Because $D = \cup_{t\in [\Lambda]}D_t$ temporally dominates all vertices in $G_{\downarrow}$, it follows that all $\cup_{C\in\mathcal{C}'} C = U$. Finally, $|D| = |\mathcal{C}'| \leq k$. 

    ($\impliedby$): Suppose now that $\mathcal{C}'$ is a \emph{yes}-certificate for $(U,\mathcal{C},k)$, with $|\mathcal{C}'|\leq k$. We claim that $D := \{(c,t): C_t\in\mathcal{C}'\}$ is a certificate for $x=(\mathcal{G},\Gamma)$ as a \emph{yes}-instance of $(1,0)$-\textsc{tds}. Since $\mathcal{C}'$ covers all elements of $U$, we know all vertices in $G_{\downarrow}$ must be temporally dominated. Moreover, we need only check the $(1,0)$-temporal dominating condition for the leaf vertices $\{v_1, \dots, v_n\}$. However, for all such $v$, $$|N_{\downarrow}(v)| = 1 = |D(v)|,$$and thus the condition holds. Because we added only one vertex $(c,t)$ for each $C_t\in\mathcal{C}'$, it follows that $|D| = |\mathcal{C}'| \leq k$. 
  \end{proof}

Finally, we obtain the main para-NP-hardness result. 

\begin{proof}[Theorem \ref{thm:a-b-p-np-hard}]
    Both $(0,1)$ and $(1,0)$-\textsc{tds} are NP-hard when the underlying graphs are stars, by \cite[Theorem 3.1]{HerrmannKomMora2025} and Lemma \ref{lem:1-0-tds-np-hard}, respectively. Moreover, \cite[Lemma 15]{EnrightHandLJMeeks2024} tells us that, whenever the underlying graph is a star, the temporal graph has TCW at most $3$. Putting these together, the result follows.
  \end{proof}

\section{Conclusion and open problems}
\label{sec:conc}
We have introduced the problem of $(\alpha,\beta)$-\textsc{Temporal Dominating Set}, specifically studying the problem with respect to structural temporal graph parameters. We have obtained a range of results showing both agreement and divergence from the corresponding static problem and parameters. We also defined the class of \emph{snapshot-set} problems, for which only the set of unique snapshots has bearing on instance outcome. 

A natural next step is to solve the remaining open cases from Table \ref{tab:resultsTable}, specifically: $(0,\beta)$-\textsc{tds} by TMW and TCW when $\beta$ is any positive constant; $(1,\beta)$-\textsc{tds} by TMW when $\beta$ is any non-positive constant; and $(1,\beta)$-\textsc{tds} by TCW, when $\beta$ is a negative constant. Additionally, as is the case with many temporal extensions of \textsc{Dominating Set}, it would be interesting to consider other variations on the definition of $(\alpha,\beta)$-\textsc{tds}. For example, changing the constraints on \emph{which} vertices must satisfy the $(\alpha,\beta)$-temporal dominating condition may lead to interesting extensions, perhaps even leading to a definition which captures \emph{both} \textsc{Temporal Dominating Set} and \textsc{Temporal Vertex Cover}. Finally, we would like to understand the behaviour of $(\alpha,\beta)$-\textsc{tds} both with respect to other parameters, and on cases with restricted underlying graphs (ie: a tree). Example parameters include \emph{vertex-interval-membership-width} \cite{BumpusMeeks2023,FamiliesTractableProblems}, or those for which \textsc{tds} is known to be in FPT, such as underlying treewidth plus maximum snapshot degree \cite{HerrmannKomMora2025}. 

\subsection*{Acknowledgements}

Jessica Enright and Kitty Meeks are supported by EPSRC grants EP/T004878/1 and EP/V032305/1. Elena Moss is supported by a Nokia-Bell Labs scholarship.

\bibliographystyle{plain}
\bibliography{ref1}
\end{document}